\documentclass[10pt]{article}

\usepackage{latexsym}
\usepackage{color,graphicx}
\usepackage{amssymb}
\usepackage{amsmath}
\usepackage{enumitem}
\usepackage{fullpage}
\usepackage{enumitem}
\usepackage{authblk}
\usepackage{amsthm}

\newtheorem{theorem}{Theorem}
\newtheorem{lemma}{Lemma}
\newtheorem{corollary}{Corollary}
\newtheorem{proposition}{Proposition}

\newcommand{\nsn}{\textup{\texttt{ns}}}  % node search number
\newcommand{\pw}[1]{\textup{\texttt{pw}}(#1)}
\newcommand{\costStrategy}[1]{\gamma(#1)}

\newcommand{\cR}{\mathcal{R}}
\newcommand{\cS}{\mathcal{S}}

\newcommand{\cP}{\mathcal{P}}
\newcommand{\cI}{\mathcal{I}}

\newcommand{\st}{\hspace{0.1cm}\bigl|\bigr.\hspace{0.1cm}}
\newcommand{\card}[1]{\left|#1\right|}
\newcommand{\reals}{\mathbb{R}}

\newcommand{\representation}{\cR}
\newcommand{\leftendpoint}[1]{\textup{left}(#1)}
\newcommand{\rightendpoint}[1]{\textup{right}(#1)}

\newcommand{\cliqueSize}[1]{\omega(#1)}     % width measure for a chordal supergraph
\newcommand{\treewidth}[1]{\textup{\texttt{tw}}(#1)}
\newcommand{\fillin}[1]{\textup{\texttt{fill-in}}(#1)}

\newcommand{\icost}{ic}    % for semi-canonical representations and graphs
\newcommand{\wid}{w}       % for interval representations and interval graphs
\newcommand{\iwid}{iw}     % for any graph
\newcommand{\profileG}[1]{p(#1)}
\newcommand{\profileFG}[1]{p_f(#1)}

\newcommand{\AlgIntervalCompletion}{\texttt{IC}}

\newcommand{\problemPP}{\texttt{PPM}}
\newcommand{\problemTF}{\texttt{TFM}}

\usepackage{algorithmic}
\usepackage{algorithm}

\floatname{algorithm}{Procedure}

\begin{document}

\title{On tradeoffs between width- and fill-like graph parameters}

\author[1]{Dariusz Dereniowski\footnote{Author partially supported by National Science Centre (Poland) grant number 2015/17/B/ST6/01887. Email: deren@eti.pg.edu.pl}}

\author[2]{Adam Sta\'{n}ski\footnote{This research has been conducted prior to the employment of this author in Amazon. Email: stanskia@amazon.com}}

\affil[1]{Faculty of Electronics, Telecommunications and Informatics, \mbox{Gda{\'n}sk University of Technology}, Gda{\'n}sk, Poland}

\affil[2]{Amazon.com, Gda{\'n}sk, Poland}

\date{}

\maketitle

\begin{abstract}
In this work we consider two two-criteria optimization problems: given an input graph, the goal is to find its interval (or chordal) supergraph that minimizes the number of edges and its clique number simultaneously.
For the interval supergraph, the problem can be restated as simultaneous minimization of the pathwidth $\pw{G}$ and the profile $\profileG{G}$ of the input graph $G$.
We prove that for an arbitrary graph $G$ and an integer $t\in\{1,\ldots,\pw{G}+1\}$, there exists an interval supergraph $G'$ of $G$ such that for its clique number it holds $\omega(G')\leq(1+\frac{2}{t})(\pw{G}+1)$ and the number of its edges is bounded by $|E(G')|\leq(t+2)\profileG{G}$.
In other words, the pathwidth and the profile of a graph can be simultaneously minimized within the factors of $1+\frac{2}{t}$ (plus a small constant) and $t+2$, respectively.
Note that for a fixed $t$, both upper bounds provide constant factor approximations.
On the negative side, we show an example that proves that, for some graphs, there is no solution in which both parameters are optimal.

In case of finding a chordal supergraph, the two corresponding graph parameters that reflect its clique size and number of edges are the treewidth and fill-in.
We obtain that the treewidth and the fill-in problems are also `orthogonal' in the sense that for some graphs, a solution that minimizes one of those parameters cannot minimize the other.
As a motivating example, we recall graph searching games which illustrates a need of simultaneous minimization of these pairs of graph parameters.
\end{abstract}


\section{Introduction} \label{sec:introduction}

Multi-criteria optimization problems can be of interest for several reasons, including theoretical insights their study provides or potential practical applications.
The selection of the parameters to be simultaneously optimized is dictated by those and can lead to challenging research questions.
Our selection is motivated in two ways.
First, the choice of the parameters themselves is made according to their importance in graph theory and algorithm design.
Second, we paired the parameters according to a potential application that we describe in detail.
The first pair of parameters that we minimize is the pathwidth and profile, which can be viewed as computations of linear graph layouts of certain characteristics.
The second pair is the treewidth and fill-in, which is a tree-like graph layout counterpart of the former.

\subsection{Related work}
We point out several optimization problems in which pathwidth or treewidth is paired with another parameter or with additional conditions that need to be satisfied.
For an example consider a problem of computing a path decomposition with restricted width and length (defined as the number of bags in the path decomposition).
It has been first studied in \cite{ADGJT11} as a problem motivated by an industrial application and called the \emph{partner unit problem} but finds applications also in scheduling a register allocation~\cite{Sethi75} or graph searching games \cite{DereniowskiKZ15}.

The treewidth counterpart of the `length' minimization problem can be defined two-fold.
It can be seen as minimizing, besides of the width of a tree decomposition, the number of its bags \cite{LiMNS15}.
On the other hand, minimization of width and the diameter of the underlying tree-structure of the decomposition has been studied in \cite{Bodlaender89,BodlaenderH98}.

Pathwidth or treewidth parameters have been also studied with additional constraints which can be most generally stated as requiring certain connectivity structures to be induced by the bags.
These include the parameter of connected pathwidth introduced in~\cite{BFFS02} in the context of graph searching games and studied further e.g. in~\cite{BarriereFFFNST12,Dereniowski12SIDMA}.
Another examples are connected treewidth \cite{connected_treewidth}, or bounded diameter tree decompositions~\cite{Bodlaender88,H98}.

Some more examples of very closely related two-criteria problems can be found in the graph searching games; see e.g.~\cite{BGHK95,DereniowskiD13,KloksBodlaender92}.

\subsection{Outline}
This work mostly deals with simultaneous minimization of width-like (namely pathwidth and treewidth) and fill-like (namely profile and fill-in) graph parameters.
In order to state our results for pathwidth and profile formally, we introduce the necessary notation in Section~\ref{sec:preliminaries}.
For pathwidth and profile we give an upper bound (to be precise, a class of upper bounds that results in a tradeoff between the two parameters) in Section~\ref{sec:tradeoff} (Theorem~\ref{thm:tradeoff}) and, in Section~\ref{sec:orthogonal}, an example that shows that the two cannot be simultaneously minimized in general (Theorem~\ref{thm:orthogonal-pw}).
The latter example also is valid for the tradeoffs between the two corresponding parameters, treewidth and fill-in and for this reason we introduce the two in Section~\ref{sec:treewidth-fill-in} and state this result as a corollary (Corollary~\ref{cor:orthogonal-tw}).
Section~\ref{sec:applications} recalls two classical graph searching problems which serve as an example that illustrates a case in which it is natural to optimize the two selected pairs of parameters.
These connections are summarized there in Theorems~\ref{thm:searching-pathwidth} and~\ref{thm:searching-treewidth}.
Thus, this part of the work serves as an additional motivation for this research.

\section{Preliminaries} \label{sec:preliminaries}
We start with recalling some basic graph-theoretic terms used in this work.
For a graph $G$, we write $V(G)$ and $E(G)$ to denote the sets of its vertices and edges, respectively.
We say that a graph $G'$ is a \emph{subgraph} of a graph $G$ (and in such case $G$ is a \emph{supergraph} of $G'$) if $V(G')\subseteq V(G)$ and $E(G')\subseteq E(G)$.
Moreover, $G'$ is a subgraph of $G$ \emph{induced} by $X\subseteq V(G)$ and denoted $G[X]$ (or $G'$ is an \emph{induced subgraph} of $G$ for short) if $V(G')=X$ and $E(G')=\{\{u,v\}\in E(G)\st u,v\in X\}$.
A \emph{clique} is a graph in which any two vertices are adjacent.
For a vertex $v$ of a graph $G$, $N_G(v)$ is the set of neighbors of $v$ in $G$.

\medskip
We now recall the graph parameters studied in this work.
For a a permutation $f\colon V(G)\to\{1,\ldots,|V(G)|\}$ of the vertices of $G$, define
\[\profileFG{G} = \sum_{v\in V(G)} \left( f(v) - \min_{u\in \{v\}\cup N_G(v)} f(u) \right).\]
Informally, $f(v) - \min_{u\in \{v\}\cup N_G(v)} f(u)$ can be interpreted as the maximum distance, according to the permutation $f$, between $v$ and its neighbors appearing in the permutation prior to $v$ (if all neighbors of $v$ are ordered in $f$ after $v$, then this difference is by definition zero).
Then, a \emph{profile} of a graph $G$ \cite{GPW74}, denoted by $\profileG{G}$, is defined as
\begin{equation} \label{eq:profile-def}
\profileG{G} = \min\left\{ \profileFG{G} \st f \textup{ is a permutation of } V(G)\right\}.
\end{equation}

A \emph{path decomposition} of a simple graph $G=(V(G),E(G))$ is a sequence $\cP=(X_1,\ldots,X_d)$, where $X_i\subseteq V(G)$ for each $i=1,\ldots,d$, and
\begin{itemize}
 \item[$\circ$] $\bigcup_{i=1,\ldots,d}X_i=V(G)$,
 \item[$\circ$] for each $\{u,v\}\in E(G)$ there exists $i\in\{1,\ldots,d\}$ such that $u,v\in X_i$,
 \item[$\circ$] for each $i,j,k$, $1\leq i\leq j\leq k\leq d$ it holds $X_i\cap X_k\subseteq X_j$.
\end{itemize}
The width of $\cP$ equals $\max_{i=1,\ldots,d}|X_i|-1$, and the \emph{pathwidth} of $G$, denoted by $\pw{G}$, is the smallest width of all path decompositions of $G$.

\subsection{Interval graphs} \label{subsec:interval}

A graph $G$ is an \emph{interval graph} if and only if for each $v\in V(G)$ there exists an interval $I_v=(l_v,r_v)$ such that for each edge $u,v\in V(G)$ it holds: $\{u,v\}\in E(G)$ if and only if $I_u\cap I_v\neq\emptyset$.
The collection $\cI=\{I_v\st v\in V(G)\}$ is called an \emph{interval representation of} $G$.
An interval representation $\cI$ of $G$ is said to be \emph{canonical} if the endpoints of $I_v$ are integers for each $v\in V(G)$ and $\{l_v\st v\in V(G)\}=\{1,\ldots,n\}$.
This in particular implies that the left endpoints are pairwise different.
%It is known that for each interval graph there exists its canonical representation \cite{??}.
Denote by $\representation(G)$ the set of all canonical interval representations of $G$.
We will write $\representation(G)$ for a graph $G$ that is not an interval graph to denote the set $\bigcup_{G'\in X}\representation(G')$, where $X$ is the set of all interval supergraphs of $G$ with the same vertex set as $G$.
If $\cI$ is an interval representation of an interval graph $G$ and $v\in V(G)$, then $\cI(v)$ denotes the interval in $\cI$ that corresponds to $v$.
For any interval $I$, we write $\leftendpoint{I}$ and $\rightendpoint{I}$ to denote its left and right endpoint, respectively.
Note that we consider without loss of generality only open intervals in the interval representations.

Let $G$ be an interval graph.
Given an interval representation $\cI$ of $G$ and an integer $i$, define
\[m_i(\cI)=\card{\{I\in\cI\st i\in I\}}\]
to be the cardinality of the set of all intervals that contain the point $i$.
Let $I_1,\ldots,I_n$ be the intervals in $\cI$. % sorted so that $\leftendpoint{I_i}\leq\leftendpoint{I_{i+1}}$ for each $i\in\{1,\ldots,n-1\}$.
Then, let $f_i(\cI)=m_j(\cI)$, where $j=\leftendpoint{I_i}$, $i\in\{1,\ldots,n\}$.
In other words, $f_i(\cI)$ is the number of intervals in $\cI$ containing the point $\leftendpoint{I_i}$.
(Note that $\leftendpoint{I_i}\notin I_i$ and hence $I_i$ does not contribute to the value of $f_i(\cI)$.)

Given a canonical interval representation $\cI$ of an interval graph $G$, define the \emph{interval cost} of $\cI$ as
\[\icost(\cI)=\sum_{i=1}^n f_i(\cI),\quad\textup{ where }n=\card{\cI}.\]
It turns out that $\icost(\cI)$ equals the number of edges of the interval graph with interval representation $\cI$ (see e.g. \cite{interval_completion}).
For a graph $G$, we define its \emph{interval cost} as
\[\icost(G)=\min\left\{\icost(\cI)\st\cI\in\representation(G)\right\}.\]

The next fact follows from \cite{Billionnet86} and \cite{interval_completion}.
\begin{proposition} \label{pro:interval-profile}
Let $G$ be any graph and let $k$ be an integer.
The following inequalities are equivalent:
\begin{enumerate}[label={\normalfont(\roman*)},leftmargin=*]
 \item $\card{E(G')}\leq k$, where $G'$ is an interval supergraph of $G$ having the minimum number of edges,
 \item $\profileG{G}\leq k$.
\end{enumerate}
\end{proposition}

Let $G$ be an interval graph and let $\cI\in\representation(G)$.
We define the \emph{width of $\cI$} as $\wid(\cI)=\max\left\{m_i(\cI)\st i\in\reals \right\}$.
%Then, the \emph{width} of and interval graph $G'$ is $\wid(G')=\min\{w(\cI)\st\cI\in\representation(G')\}$.
The \emph{interval width} of any simple graph $G$ is then
\[\iwid(G)=\min\left\{\wid(\cI)\st \cI\in\representation(G)\right\}.\]

We have the following fact \cite{Kinnersley92,KirousisPapadimitriou85,searching_and_pebbling,LaPaugh93,Mohring90}:
\begin{proposition} \label{pro:interval-pathwidth}
Let $G$ be any graph and let $k$ be an integer.
The following inequalities are equivalent:
\begin{enumerate}[label={\normalfont(\roman*)},leftmargin=*]
 \item $\iwid(G)\leq k$,
 \item $\pw{G}\leq k-1$.
\end{enumerate}
\end{proposition}

\subsection{Problem formulation} \label{sec:problem_formulation}
For the purposes of this work we need an `uniform' formulation of the two graph-theoretic problems that we study, namely pathwidth and profile, in order to be able to formally apply the two optimization criteria to a single solution to a problem instance.
In view of Propositions~\ref{pro:interval-profile} and~\ref{pro:interval-pathwidth}, we can state the optimization version of our problem as follows:

\medskip
\noindent
\textbf{Problem $\problemPP$} (\emph{Pathwidth \& Profile Minimization}):
\begin{description}[labelindent=\parindent,noitemsep,topsep=0pt]
   \item[{\normalfont Input:}] a graph $G$, integers $k$ and $c$.
   \item[{\normalfont Question:}] does there exist an interval supergraph $G'$ of $G$ such that $\iwid(G')\leq k$ and $\card{E(G')}\leq c$?
\end{description}

\section{Pathwidth and profile tradeoffs} \label{sec:tradeoff}

In this section we prove that for an arbitrary graph $G$ there exists its interval supergraph $G'$ with width at most $(1+\frac{2}{t})(\pw{G}+1)$ and the number of edges at most $(t+2)\profileG{G}$ for each $t\in\{1,\ldots,\iwid(G)\}$.
This is achieved by providing a procedure that finds a desired interval supergraph (the procedure returns an interval representation of this supergraph).
Since the goal is to prove an upper bound and not to provide an efficient algorithm, this procedure relies on optimal algorithms for finding a minimum width and minimum cost interval supergraph of a given graph.
(The latter problems are NP-complete, see \cite{KashiwabaraF79,Lengauer81} and \cite{GareyJohnson79,YuanLLW98}.)
Therefore, the running time of this procedure is exponential.

\medskip
We first give some intuitions on our method.
We start by computing a `profile-optimal' (canonical) interval representation $\cI$ of some interval supergraph $G''$ of $G$, that is, $\icost(\cI)=\profileG{G}$.
Then, in the main loop of the procedure this initial interval representation is iteratively refined.
Each refinement targets an interval $(i,j)$ selected in such a way that the width of $\cI$ exceeds $k/t$ at each point in $(i,j)$, i.e., $m_{i'}(\cI'')>k/t$ for each $i'\in\{i,\ldots,j\}$ and $i,j$ are taken so that the interval is maximum with respect to this condition.
The refinement on $\cI$ in $(i,j)$ is done as follows (see the pseudocode below for detail and Figure~\ref{fig:interval} for an example):
\begin{itemize}
 \item intervals that cover $(i,j)$ entirely or have an empty intersection with it do not change (see Case~(i) in Figure~\ref{fig:interval}),
 \item intervals that contain one of $i$ or $j$ will be extended to cover entire interval, except that we ensure that they have pairwise different left endpoints as required in canonical representations (Cases~(ii) and~(iii) in Figure~\ref{fig:interval}),
 \item for the intervals that originally are contained in $(i,j)$, we recompute the interval representation; while doing so we take care of the following: first, the neighborhood relation in the initial graph is respected so that the new interval representation provides an interval supergraph of $G$, and second, the width of the new interval representation inside $(i,j)$ is minimal (Case~(iv) in Figure~\ref{fig:interval}).
\end{itemize}
The above refinement is performed for each interval $(i,j)$ that satisfies given conditions.
Each refinement potentially increases the interval cost of $\cI$ but narrows down its width appropriately in the interval for which the refinement is done.

\medskip
We now describe Procedure~$\AlgIntervalCompletion$ (\emph{Interval Completion}) that as an input takes any graph $G$ and an integer $t\in\{1,\ldots,\iwid(G)\}$, and returns an $\cI'\in\representation(G)$.
Then, in Lemma~\ref{lem:IC_is_correct}, we prove that the procedure is correct and in Theorem~\ref{thm:tradeoff} we estimate the width and interval cost of $\cI'$.
See Figure~\ref{fig:interval} for an example that illustrates the transition from the $\cI''$ to $\cI'$ in an iteration of the procedure.
\begin{figure}[htb]
\begin{center}
\centering
  \includegraphics[scale=0.8]{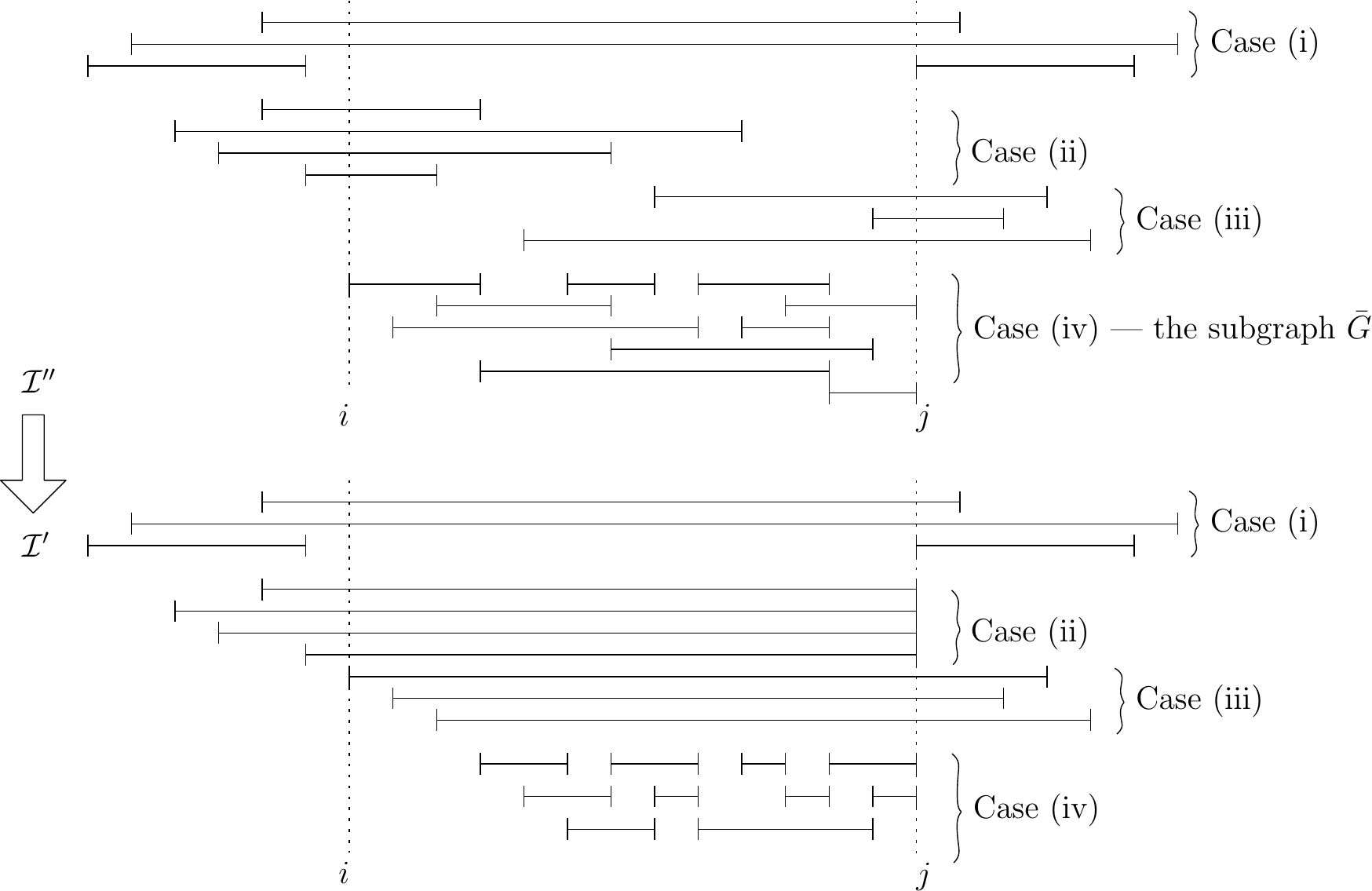}
\end{center}
\caption{A single iteration of Procedure~$\AlgIntervalCompletion$: transition from interval representation $\cI''$ to $\cI'$ with particular intervals marked according in which of the cases (i)-(iv) they are processed in the pseudocode}
\label{fig:interval}
\end{figure}

\begin{algorithm} \caption{$\AlgIntervalCompletion$ (\emph{Interval Completion})}
\begin{algorithmic}
  \REQUIRE A graph $G$ and an integer $t\in\{1,\ldots,\iwid(G)\}$.
  \ENSURE An interval representation of some interval supergraph of $G$.

  \STATE Compute an interval representation $\cI\in\representation(G)$, where $\icost(\cI)=\profileG{G}$.
  \STATE Set $k\leftarrow\iwid(G)$, $q\leftarrow 1$, $\cI'\leftarrow\cI$ and $\cI''\leftarrow\cI$.
  \WHILE{$m_{i'}(\cI'')>k/t$ for some $i'\geq q$}
    \STATE Find minimum integers $i,j$ such that $q\leq i<j$, $m_{i-1}(\cI'')\leq k/t$, $m_{j+1}(\cI'')\leq k/t$ and $m_{i'}(\cI'')>k/t$ for each $i'\in\{i,\ldots,j\}$.
    \STATE Set $q\leftarrow j+1$.
    \STATE Let $\bar{G}$ be the subgraph of $G$ induced by all vertices $v$ such that $\cI''(v)\subseteq(i,j)$.
    \STATE Let $\bar{\cI}$ be a minimum width canonical interval representation of some interval supergraph of $\bar{G}$.
    \STATE Construct an interval representation $\cI'$ of $G$ as follows:
    \STATE\hspace{\algorithmicindent} (i) Let initially $\cI'=\{I\in\cI''\st I\cap (i,j)=\emptyset\textup{ or }(i,j)\subsetneq I\}$.
    \STATE\hspace{\algorithmicindent} (ii) For each $v\in V(G)$ such that $i\in\cI''(v)$ and $j\notin\cI''(v)$, add $\cI'(v)\leftarrow(\leftendpoint{\cI''(v)},j)$ to $\cI'$.
    \STATE\hspace{\algorithmicindent} (iii) Let $v_1,\ldots,v_p$ be all vertices of $G$ such that $i\notin\cI''(v_s)$ and $j\in\cI''(v_s)$ for each $s\in\{1,\ldots,p\}$.
    \STATE\hspace{\algorithmicindent} \phantom{(iii)} For each $s\in\{1,\ldots,p\}$, add $\cI'(v_s)\leftarrow(i+s-1,j)$ to $\cI'$.
    \STATE\hspace{\algorithmicindent} (iv) For each $v\in V(\bar{G})$, set $\cI'(v)\leftarrow(i+p+\leftendpoint{\bar{\cI}(v)}-1,i+p-2+\rightendpoint{\bar{\cI}(v)})$.
    \STATE Set $\cI''\leftarrow\cI'$.
  \ENDWHILE
  \RETURN $\cI'$
\end{algorithmic}

\end{algorithm}

\begin{lemma} \label{lem:IC_is_correct}
Let $G$ be any graph and let $t\in\{1,\ldots,\iwid(G)\}$.
Procedure~$\AlgIntervalCompletion$ for the given $G$ and $t$ returns a canonical interval representation of some interval supergraph of $G$.
\end{lemma}
\begin{proof}
Let $\cI$ be the canonical interval representation of some interval supergraph of $G$ computed at the beginning of Procedure~$\AlgIntervalCompletion$.
We proceed by induction on the number of iterations of the main loop of Procedure~$\AlgIntervalCompletion$, namely, we prove that the interval representation $\cI''$ obtained in the $s$-th iteration is a canonical interval representation of some interval supergraph of $G$.
For the purpose of the proof, we use the symbol $\cI_s$ to denote the interval representation obtained in the $s$-th iteration, taking $\cI_0=\cI$.

For the base case of $s=0$ we have that $\cI_0=\cI$ and the claim follows.
Hence, let $s>0$.
Since $\cI_s$ consists of $\card{V(G)}$ intervals, $\cI_s$ is an interval representation of an interval graph $G'$ on $\card{V(G)}$ vertices.
We need to prove that $G$ is a subgraph of $G'$ and that $\cI_s$ is canonical.

By the induction hypothesis, there exists an interval supergraph $G''$ of $G$ and $\cI_{s-1}\in\representation(G'')$.
Note that $V(G)=V(G')=V(G'')$.

To prove that $G$ is a subgraph of $G'$, we argue that $N_{G}(v)\subseteq N_{G'}(v)$ for each $v\in V(G)$.
Denote
\[X_1=\left\{x\in V(G)\st i\in\cI_{s-1}(x)\textup{ and }j\notin\cI_{s-1}(x)\right\} \textup{ and }
  X_2=\left\{x\in V(G)\st i\notin\cI_{s-1}(x)\textup{ and }j\in\cI_{s-1}(x)\right\},
\]
where $i$ and $j$ have values as in the $s$-th iteration.
Also, $\bar{G}$ refers to the subgraph computed in the $s$-th iteration.
For $v\in V(G)\setminus(X_1\cup X_2\cup V(\bar{G}))$ we have that $N_{G'}(v)=N_{G''}(v)$ and hence, since $G''$ is a supergraph of $G$, $N_G(v)\subseteq N_{G''}(v)$ gives the claim.
If $v\in X_1\cup X_2$, then
\begin{eqnarray}
N_{G''}(v) & =         & \left(N_{G''}(v)\cap(X_1\cup X_2\cup V(\bar{G}))\right)\cup\left(N_{G''}(v)\setminus(X_1\cup X_2\cup V(\bar{G}))\right) \nonumber \\
         & \subseteq & \left(X_1\cup X_2\cup V(\bar{G})\right)\cup\left(N_{G''}(v)\setminus(X_1\cup X_2\cup V(\bar{G}))\right) =N_{G'}(v). \nonumber
\end{eqnarray}
Thus, for $v\in X_1\cup X_2$ we also have $N_G(v)\subseteq N_{G''}(v)\subseteq V_{G'}(v)$ as required.
If $v\in V(\bar{G})$, then
\[N_{G}(v)\subseteq N_{\bar{G}}(v)\cup X_1\cup X_2=N_{G'}(v).\]

Now we argue that $\cI_s$ is canonical.
Since both endpoints of each interval in $\cI_s$ are clearly integers, it is sufficient to prove that for each $i'\in\{1,\ldots,\card{V(G)}\}$ there exists exactly one interval $I\in\cI_s$ whose left endpoint equals $i'$.
If $i'\in\{1,\ldots,i-1\}\cup\{j,\ldots,\card{V(G)}\}$, then the claim follows from
\[\left\{v\in V(G)\st (i',i'+1)\subseteq\cI_{s-1}(v)\right\}=\left\{v\in V(G)\st (i',i'+1)\subseteq\cI_s(v)\right\},\]
i.e., the interval representations $\cI_{s-1}$ and $\cI_s$ are identical `outside' of the interval $(i,j)$.
For each $i'\in\{i,\ldots,i+p-1\}$ the claim holds, since $\leftendpoint{\cI_s(v_{i'-i+1})}=i'$.
Finally, let $i'\in\{i+p,\ldots,j-1\}$.
Since, $\cI_{s-1}$ is canonical, $\card{V(\bar{G})}=j-i-p-1$.
This implies that $\{\leftendpoint{\cI}\st\cI\in\bar{\cI}\}=\{1,\ldots,\card{V(\bar{G})}\}$ because $\bar{\cI}$ is canonical.
Thus, there exists $v\in V(\bar{G})$ such that $\leftendpoint{\bar{\cI}(v)}=i'-i-p+1$.
By the construction of $\cI_s$, $\leftendpoint{\cI_s(v)}=i'$ as required.
\end{proof}

\begin{theorem} \label{thm:tradeoff}
Let $G$ be any graph and let $t\in\{1,\ldots,\pw{G}+1\}$ be an integer.
There exists an interval supergraph $G'$ of $G$ and $\cI'\in\representation(G')$ such that $\wid(\cI')\leq(1+\frac{2}{t})(\pw{G}+1)$ and $\icost(\cI')\leq(t+2)\profileG{G}$.
\end{theorem}
\begin{proof}
Suppose that Procedure~$\AlgIntervalCompletion$ is executed for the input $G$ and $t$.
Let $\cI$ be the canonical interval representation of some interval supergraph of $G$ computed at the beginning of Procedure~$\AlgIntervalCompletion$.
Moreover, take such an $\cI$ that satisfies $\icost(\cI)=\icost(G)$.
Let $r$ be the number of iterations performed by the main loop.
Let $\bar{G}_q$ and $\bar{\cI}_q$, $q\in\{1,\ldots,r\}$, be the graph $\bar{G}$ and its interval representation, respectively, computed in the $q$-th iteration of the main loop.
Also, let $(i_q,j_q)$ be the interval used to select $\bar{G}_q$, i.e., $\bar{G}_q$ is the subgraph of $G$ induced by all vertices $v$ such that $\cI''(v)\subseteq(i_q,j_q)$ for each $q\in\{1,\ldots,r\}$.

Note that Procedure~$\AlgIntervalCompletion$ does not specify how $\bar{\cI}_q$ is selected for each $q\in\{1,\ldots,r\}$ and therefore for the purpose of this proof of upper bounds we may assume that the interval representation $\bar{\cI}_q$ satisfies
\begin{equation} \label{eq:wid_bar_cI}
\wid(\bar{\cI}_q)=\iwid(\bar{G}_q).
\end{equation}
By Lemma~\ref{lem:IC_is_correct}, $\cI'$ returned by Procedure~$\AlgIntervalCompletion$ is a canonical representation of some interval supergraph of $G$.
By construction,
\begin{equation} \label{eq:in_gaps}
m_p(\cI')=m_p(\cI) \textup{ for each }p\notin\bigcup_{q=1}^r\{i_q,\ldots,j_q\}.
\end{equation}
Since $\bar{G}_q$ is a subgraph of $G$, $\iwid(\bar{G}_q)\leq\iwid(G)$ for each $q\in\{1,\ldots,r\}$ and hence by~\eqref{eq:wid_bar_cI} we obtain $\wid(\bar{\cI}_q)\leq\iwid(G)$.
By the choice of $i_q$ and $j_q$, we have $m_{i_q-1}(\cI)\leq \iwid(G)/t$ and $m_{j_q+1}(\cI)\leq \iwid(G)/t$.
Hence, we obtain
\begin{equation} \label{eq:in_intervals1}
m_p(\cI')\leq m_{i_q-1}(\cI)+m_{j_q+1}(\cI)+\wid(\bar{\cI}_q) \leq \left(1+\frac{2}{t}\right)\iwid(G)
\end{equation}
for each $p\in\{i_q,\ldots,j_q\}$ and for each $q\in\{1,\ldots,r\}$.
Equations~\eqref{eq:in_gaps} and~\eqref{eq:in_intervals1} give that $\wid(\cI')\leq(1+\frac{2}{t})\iwid(G)$.
Observe that, by~\eqref{eq:in_intervals1}, for each $q\in\{1,\ldots,r\}$ and for each $p\in\{i_q,\ldots,j_q\}$ it holds
\begin{equation} \label{eq:in_intervals2}
m_p(\cI')\leq (t+2)m_p(\cI)
\end{equation}
because $m_p(\cI)\geq\iwid(G)/t$.
By~\eqref{eq:in_gaps} and~\eqref{eq:in_intervals2}, $\icost(\cI')\leq (t+2)\icost(\cI)$.
Finally note that by Proposition~\ref{pro:interval-profile}, $\icost(\cI)=\profileG{G}$ and by Proposition~\ref{pro:interval-pathwidth}, $\iwid(G)=\pw{G}+1$, which completes the proof.
\end{proof}

\section{Pathwidth and profile are `orthogonal'} \label{sec:orthogonal}
In this section we prove that the two optimization criteria studied in this work cannot be minimized simultaneously for some graphs.
In other words, we prove by example, that there exist graphs $G$ such that any interval supergraph $G'$ of $G$ that has the minimum number of edges (i.e., $\card{E(G')}=\profileG{G}$) cannot have minimum width (i.e., $\iwid(G)>\pw{G}-1$) and vice versa.
The example that we construct will be also used in the next section and for this reason we present it here in terms of chordal graphs, which is a class of graphs that generalizes interval graphs.
For that we need some additional definitions.

\medskip
We say that $C$ is an \emph{induced cycle of length $k\geq 3$} in a graph $G$ if $C$ is a subgraph of $G$ and $\{\{u,v\}\in E(G)\st u,v\in V(C)\}=V(C)$, i.e., the only edges in $G$ between vertices in $V(C)$ are the ones in $E(C)$.
A graph is \emph{chordal} if there is no induced cycle of length greater than $3$ in $G$.
Any edge that does not belong to a cycle $C$ and connects two vertices of $C$ is called a \emph{chord} of $C$.

In this section we consider a graph $G$ with vertex set $V(G)=A\cup B\cup B'\cup C$ and edges placed in such a way that $A\cup B$, $A\cup B'$, $B\cup C$ and $B'\cup C$ form cliques.
In our construction we take any sets that satisfy
\begin{equation} \label{eq:ABC}
\card{A}<\card{B}=\card{B'}<\card{C} \quad\textup{and}\quad \card{A}\card{C}>\card{B}^2.
\end{equation}

We have the following observation.
\begin{lemma} \label{lem:AllOrNothing}
If $G'$ is a minimal chordal supergraph of $G$, then each of the subgraphs $G'[A\cup C]$ or $G'[B\cup B']$ is either a clique or an union of two disconnected cliques.
\end{lemma}
\begin{proof}
We prove that the subgraph of $G'$ induced by $A\cup C'$ is either a clique or is disconnected and the proof for $B\cup B'$ is identical due to the symmetry.
If $A\cup C'$ induces a clique, then the claim follows so suppose that there exist two vertices $a\in A\cup C$ and $c\in A\cup C$ that are not adjacent in $G'$.
Without loss of generality let $a\in A$ and $c\in C$ --- this is due to the fact that $G'[A]$ and $G'[C]$ are cliques.
Take any two vertices $b\in B$ and $b'\in B'$.
Since $G'$ is chordal, the cycle induced by $a,b,b'$ and $c$ has a chord in $G'$.
Thus, there is an edge between $b$ and $b'$ in $G'$.
Since $b$ and $b'$ are selected arbitrarily, $G'[B\cup B']$ is a clique.
Note that a supergraph of $G$ that has no edge between any vertex in $A$ and any vertex in $C$ and in which $B\cup B'$ induces a clique is chordal.
Thus, by the minimality of $G'$, $G'[A\cup C]$ consists of two cliques $G'[A]$ and $G'[C]$ with no edges between them, as required.
\end{proof}

\begin{theorem} \label{thm:orthogonal-pw}
There exists a graph $G$ such that no interval supergraph $G'$ of $G$ satisfies $\iwid(G')=\iwid(G)$ and $\icost(G')=\icost(G)$.
\end{theorem}
\begin{proof}
Consider the graph $G=(A\cup B\cup B'\cup C,E(G))$ constructed at the beginning of this section.
For any minimal chordal supergraph $G'$ of $G$, we say that it is $(A,C)$-connected ($(B,B')$-connected) if $G'[A\cup C]$ is a clique ($G'[B\cup B']$ is a clique, respectively).
Denote by $G_{(A,C)}$ (respectively, $G_{(B,B')}$) the minimal chordal supergraph of $G$ that is $(A,C)$-connected ($(B,B')$-connected, respectively) but has no edge joining a vertex in $B$ (respectively $A$) with a vertex in $B'$ (respectively $C$).

Each interval graph is also chordal.
On the other hand, both $G_{(A,C)}$ and $G_{(B,B')}$ are interval graphs.

Consider a minimal chordal supergraph $G'$ of $G$.
By Lemma~\ref{lem:AllOrNothing}, $G'$ is $(A,C)$-connected, $(B,B')$-connected, both $(A,C)$- and $(B,B')$-connected or neither $(A,C)$- and $(B,B')$-connected.
Since it is minimal, it is not $(A,C)$- and $(B,B')$-connected simultaneously.
Also, it must be $(A,C)$- or $(B,B')$-connected for otherwise it is not chordal.
This implies that $G'$ equals either to $G_{(A,C)}$ or $G_{(B,B')}$.
We have that $G'$ is an interval graph and hence, by \eqref{eq:ABC}, we obtain
\[\iwid(G_{(A,C)}) - \iwid(G_{(B,B')}) = \left(\card{A}+\card{B}+\card{C}\right) - \left(\card{B}+\card{B'}+\card{C}\right) < 0,\]
and
\[\icost(G_{(A,C)}) - \icost(G_{(B,B')}) = \card{A}\card{C} - \card{B}\card{B'}>0.\]
We conclude by noting that it is enough to consider minimal interval supergraphs when minimizing interval cost or interval width.
\end{proof}

\section{Treewidth and fill-in} \label{sec:treewidth-fill-in}

We refer the reader e.g. to \cite{Proskurowski87,GareyJohnson79,RobertsonS86,Yannakakis81} for definitions of the NP-complete problems of treewidth and fill-in.
The treewidth for a given graph $G$, denoted by $\treewidth{G}$, can be defined as the the minimum $k$ such that there exists a chordal supergraph $G'$ of $G$ such that the maximum clique $\cliqueSize{G'}$ of $G'$ has size at most $k+1$.
The \emph{fill-in} of $G$ is the minimum $m$ such that there exists a chordal supergraph of $G$ that can be constructed by adding $m$ edges to $G$.
Hence, our corresponding combinatorial problem can be stated as follows:

\medskip
\noindent
\textbf{Problem $\problemTF$} (\emph{Treewidth \& Fill-in Minimization}):
\begin{description}[labelindent=\parindent,noitemsep,topsep=0pt]
   \item[{\normalfont Input:}] a graph $G$, integers $k$ and $c$.
   \item[{\normalfont Question:}] does there exist a chordal supergraph $G'$ of $G$ such that $\cliqueSize{G'}\leq k$ and $\card{E(G')}\leq c$?
\end{description}
\medskip

By the same proof as in Theorem~\ref{thm:orthogonal-pw}, we obtain that for some graphs there is no solution to Problem~$\problemTF$ in which $k=\treewidth{G}-1$ and $c=\card{E(G)}+\fillin{G}$.
\begin{corollary} \label{cor:orthogonal-tw}
There exists a graph $G$ such that no chordal supergraph $G'$ of $G$ satisfies $\treewidth{G}=\cliqueSize{G'}-1$ and $\fillin{G}+\card{E(G)}=\card{E(G')}$.
\qed
\end{corollary}

\section{Applications to graph searching} \label{sec:applications}

\subsection{A Short Introduction to Graph Searching} \label{subsec:graphsearching}
The problem of graph searching can be informally stated as one in which an agent called the \emph{fugitive} is moving around the graph with the goal to escape a group of agents called guards or \emph{searchers}.
There are various statements of this problem specifying behavior of fugitive and searchers, phase restrictions, speeds of both parties or their other capabilities like visibility, radius of capture etc.
Numerous optimization criteria have been studied for these games.
However, the tradeoffs between different optimization parameters have not yet been throughly analyzed.
In this work we refer to one of the two classical formulations of the graph searching problem, namely the \emph{node search}; see a formal definition below.

In the original statement of the problem the fugitive is considered invisible (i.e., the searchers can deduce its potential locations only based on the history of their moves) and \emph{active}, i.e., constantly moving with unbounded speed to counter the searchers' strategy.
It turns out that the minimum number of searchers sufficient to guarantee the capture of the fugitive corresponds to the pathwidth of the underlying graph~\cite{DendrisKT97}.
Later on, the lazy, also referred to as \emph{inert}, fugitive variant has been defined in which the fugitive only moves when the searchers are one move apart from catching it.
The latter version was first introduced in~\cite{DendrisKT97} where the authors show that minimizing the number of searchers precisely corresponds to finding the treewidth of the input graph.
Seymour and Thomas proposed in \cite{SeymourThomas93} a variant of the game in which the fugitive was visible and active. In the same paper they prove that the visible active variant of the problem is equivalent to the invisible inert variant.

All previously mentioned problems considered the number of searchers used by a strategy to be the optimization criterion.
In~\cite{interval_completion}, the authors analyzed the cost defined (informally) as the sum of the guard counts over all steps of the strategy.
This graph searching parameter is the one that corresponds to the profile minimization.

Not much is known in terms of two-criteria optimization in the graph searching games.
To mention some, examples, there is an analysis of simultaneous minimization of time (number of `parallel' steps) and the number of searchers for the visible variant~\cite{DereniowskiKZ15} and for the inert one~\cite{LiMNS15} of the node search.
Also, some results on tradeoffs between the cost and the number of searchers for the edge search game can be found in~\cite{DereniowskiD13}.

\subsection{Formal definitions} \label{subsec:search-defs}

The following definitions of the \emph{node search problem} are taken from or based on \cite{interval_completion} and \cite{FominHT04}.
An \emph{active search strategy} $\cS$ for a graph $G$ is a sequence of pairs
\[(A_0,Z_0),(A_1,Z_1),\ldots,(A_m,Z_m)\]
that satisfies the following \emph{axioms}:
\begin{enumerate}[label={\normalfont(\roman*)},leftmargin=*]
 \item\label{axiom:1} $A_i\subseteq V(G)$ and $Z_i\subseteq V(G)$ for each $i\in\{0,\ldots,m\}$,
 \item\label{axiom:2} $A_0=Z_0=\emptyset$, $A_m=V(G)$ and $Z_m=\emptyset$,
 \item\label{axiom:3} (\emph{placing/removing searchers}) For each $i\in\{1,\ldots,m\}$ there exist $v_i\in V(G)$ such that $\{v\}=A_i\setminus A_{i-1}$, $v_i\in Z_i$ and $Z_i\subseteq A_{i-1}\cup\{v_i\}$.
 \item\label{axiom:4} (\emph{possible recontamination}) For each $i\in\{1,\ldots,m\}$, $A_i$ consists of $v_i$ and each vertex $u$ such that each path connecting $u$ to a vertex in $V(G)\setminus A_{i-1}$ has an internal vertex in $Z_{i-1}$.
\end{enumerate}
An \emph{inert search strategy} $\cS$ is one that satisfies axioms \ref{axiom:1},\ref{axiom:2},\ref{axiom:3} and:
\begin{enumerate}[label={\normalfont(\roman*')},leftmargin=*]
\setcounter{enumi}{3}
 \item\label{axiom:4'} (\emph{possible recontamination}) For each $i\in\{1,\ldots,m\}$, $A_i$ consists of $v_i$ and each vertex $u$ such that each path connecting $u$ to $v_i$ has an internal vertex in $Z_{i-1}$.
\end{enumerate}
We say that, in the $i$-th step of $\cS$, the vertices in $Z_i$ are \emph{guarded}, the vertices in $A_i$ are \emph{cleared} and the vertices in $V(G)\setminus A_i$ are \emph{contaminated}.
The \emph{search cost} of a search strategy $\cS$ is defined as
\[\costStrategy{\cS}=\sum_{i=0}^m \card{Z_i}\]
and the number of searchers it uses is
\[\nsn(\cS)=\max\left\{\card{Z_i}\st i\in\{0,\ldots,m\}\right\}.\]
Informally speaking, in an active search strategy recontamination can `spread' from any contaminated vertex while in inert strategies recontamination can only spread from $v_i$.
% Then, the \emph{active search cost} of a graph $G$ is
% \[\costActive{G}=\min\left\{\costStrategy{\cS}\st\cS\emph{ is an active search strategy for }G\right\},\]
% the \emph{inert search cost} of a graph $G$ is
% \[\costInert{G}=\min\left\{\costStrategy{\cS}\st\cS\emph{ is an inert search strategy for }G\right\}.\]

We say that a strategy (active or inert) $\cS=((A_0,Z_0),\ldots,(A_m,Z_m))$ is \emph{monotone} if $A_i\subseteq A_{i+1}$ for each $i\in\{1,\ldots,m-1\}$.
% We similarly define the \emph{monotone active search cost} of $G$ as
% \[\costActiveMonotone{G}=\min\left\{\costStrategy{\cS}\st\cS\emph{ is an active monotone search strategy of }G\right\}\]
% and the \emph{monotone inert search cost} of $G$ as
% \[\costInertMonotone{G}=\min\left\{\costStrategy{\cS}\st\cS\emph{ is an inert monotone search strategy of }G\right\}.\]

\subsection{Consequences of our results} \label{sec:consequences}

We have the following equivalences:
\begin{theorem}[\cite{interval_completion}]
For each graph $G$, if $\cS$ an active monotone search strategy of minimum cost, then $\costStrategy{\cS}=\icost(G)$.
\qed
\end{theorem}
\begin{theorem}[\cite{Kinnersley92,KirousisPapadimitriou85,searching_and_pebbling,Mohring90}]
For each graph $G$, if $\cS$ an active search strategy that uses the minimum number of searchers, then $\nsn(\cS)=\pw{G}+1$.
\qed
\end{theorem}
Hence we obtain the following equivalence:
\begin{theorem} \label{thm:searching-pathwidth}
An optimal solution to Problem~$\problemPP$ corresponds to an active search strategy that simultaneously minimizes the number of searchers and the cost.
\qed
\end{theorem}

For the second pair of parameters, we recall the following theorems.
\begin{theorem}[\cite{FominHT04}]
For each graph $G$, if $\cS$ an inert monotone search strategy of minimum cost, then $\costStrategy{\cS}=\card{E(G)}+\fillin{G}$.
\qed
\end{theorem}
\begin{theorem}[\cite{SeymourThomas93}]
For each graph $G$, if $\cS$ an inert search strategy that uses the minimum number of searchers, then $\nsn(\cS)=\treewidth{G}+1$.
\qed
\end{theorem}
This leads us to the following theorem:
\begin{theorem} \label{thm:searching-treewidth}
An optimal solution to Problem~$\problemTF$ corresponds to an inert search strategy that simultaneously minimizes the number of searchers and the cost.
\qed
\end{theorem}

\section{Open problems} \label{sec:open-problems}

The first open problem we leave is the one of existence of a similar tradeoff between fill-in and treewidth to the one we have in Theorem~\ref{thm:tradeoff}.
More particularly, is it possible to find chordal supergraphs that approximate both parameters to within constant factors of their optimal values?

A challenging open problem is the one that refers to the concept of recontamination in the graph searching games that has been posed in~\cite{FominHT04}: does recontamination help to obtain a minimum-cost inert search strategy?
Formally, does there exist, for some graph $G$, an inert search strategy whose cost is smaller than $\card{E(G)}+\fillin{G}$?
In yet other words, does there exist a graph for which an inert search strategy that minimizes the cost must necessarily allow for recontamination and as a result some vertex $v$ is searched twice in step \ref{axiom:3}, i.e., $v=v_i$ for two different indices $i$?

\bibliographystyle{plain}
\bibliography{search}
\end{document}